\documentclass[aps,12pt,a4paper]{revtex4-2}
\usepackage[english]{babel}

\usepackage{graphicx} 
\usepackage{bm} 
\usepackage{amssymb} 
\usepackage[inline]{enumitem} 
\usepackage{mathtools}
\usepackage[caption=false]{subfig} 
\usepackage{amsthm} 
\usepackage[dvipnames]{xcolor} 

\parskip=\medskipamount
\setlist[itemize]{noitemsep,label=\small$\blacktriangleright$}
\setlist[enumerate]{noitemsep}

\DeclarePairedDelimiterX\innp[2]{\langle}{\rangle}{#1,#2}

\DeclarePairedDelimiter\floor{\lfloor}{\rfloor}

\DeclareMathOperator{\EX}{\mathbb{E}}
\DeclareMathOperator{\prob}{\mathbb{P}}

\newcommand{\N}{{\mathbb{N}}}

\newcommand{\diff}{\mathop{}\!\mathrm{d}}

\let\phi\varphi

\let\oldforall\forall
\renewcommand{\forall}{\oldforall \, }
\let\oldexist\exists
\renewcommand{\exists}{\oldexist \: }

\newcommand{\be}{\begin{equation}}
\newcommand{\ee}{\end{equation}}

\newcommand\eps{\varepsilon}

\newcommand{\eq}[1]{(\ref{#1})} 
\newcommand{\fig}[1]{Fig.~\ref{#1}}

\newcommand{\lev}[1]{\lambda_{\max#1}}
\newcommand{\levtri}[1]{\lambda_{\max#1}}
\newcommand{\levlin}[1]{\lambda_{\max#1}^{\rm(lin)}}
\newcommand{\levbound}[1]{\lambda_{b#1}}
\newcommand{\varnlin}[1]{n_{\rm lin}}
\newcommand{\dmax}[1]{d_{\max#1}}
\newcommand{\eignorm}[1]{{\cal C}}

\newcommand{\ERname}{Erd\H{o}s-R\'enyi}

\newcommand{\fancyG}{{\cal G}}
\DeclareMathOperator{\LSR}{LSR}

\newtheorem{proposition}{Proposition}

\begin{document}



\title{Largest eigenvalue statistics of sparse random adjacency matrices}

\author{Bogdan Slavov}
\author{Kirill Polovnikov}
\affiliation{Skolkovo Institute of Science and Technology, 121205 Moscow, Russia} 

\author{Sergei Nechaev} 
\affiliation{LPTMS, Universit\'e Paris Saclay, 91405 Orsay Cedex, France}
\affiliation{Laboratory of Complex Networks, Brain and Consciousness Research Center, Moscow, Russia}

\author{Nikita Pospelov}
\affiliation{Institute of Chemical Physics, Moscow, Russia}
\affiliation{Moscow State University, Moscow, Russia}


\begin{abstract}

We investigate the statistics of the largest eigenvalue, $\lambda_{\rm max}$, in an ensemble of $N\times N$ large ($N\gg 1$) sparse  adjacency matrices, $A_N$. The most attention is paid to the distribution and typical fluctuations of $\lambda_{\rm max}$ in the vicinity of the percolation threshold, $p_c=\frac{1}{N}$. The overwhelming majority of subgraphs representing $A_N$ near $p_c$ are exponentially distributed linear subchains, for which the statistics of the normalized largest eigenvalue can be analytically connected with the Gumbel distribution. For the ensemble of {\rm all} subgraphs near $p_c$ we suggest that under an appropriate modification of the normalization constant the Gumbel distribution provides a reasonably good approximation. Using numerical simulations we demonstrate that the proposed transformation of $\lambda_{\rm max}$ is indeed Gumbel-distributed and the leading finite-size corrections in the vicinity of $p_c$ scale with $N$ as $\sim \ln^{-2}N$. All together, our results reveal a previously unknown universality in eigenvalue statistics of sparse matrices close to the percolation threshold. 


\end{abstract}


\maketitle

\tableofcontents


\section{Introduction}

The rare-event statistics has many manifestations in natural sciences. To name but a few, we can mention the peculiar statistics of communication receivers \cite{planat}, of sparse contact maps of protein-protein interactions \cite{fly}, of individual DNA molecules in cell nuclei \cite{cell}. The peculiarity of thermal noise emerges on the level of nano-objects \cite{nano1,nano2}, in dynamic properties of dendritic polymers \cite{dendritic}, etc. The everyday experience tells us that it is difficult to expect a nontrivial statistical patterns in sparse datasets. However, the spectral analysis of sparse datasets often demonstrates very peculiar hierarchically organized patterns. The rare-event statistics naturally emerges in high dimensional spaces, where it manifests itself in the special hierarchical organization of distances between points, known as ``ultrametricity'' \cite{ultram}. Sparse statistics and ultrametricity together are rooted in high dimensionality and randomness. That has been unambiguously shown in \cite{zubarev}, where it was proved that in a $D$-dimensional Euclidean space the distances between points in a highly sparse samplings tend to the ultrametric distances as $D\to\infty$.

From the other hand, experimenting with physical properties of highly diluted solutions of biologically active substances, one should pay attention to a very peculiar structure of a background noise originating from the rare-event statistics of dissolved clusters. The peculiar shape of a sparse random noise spectrum can be misinterpreted, or at least can make the data incomprehensible \cite{mairal2014sparse, peleg2010exploiting}. In order to conclude about any biological activity of regarded substance, the signal from background noise should be clearly identified. From this point of view, the work \cite{cell} seems very interesting, since it represents an exceptional example of careful attention of to unusual hierarchical distributions in real biological and clinical data which are indebted to randomness.

The information about topological and statistical properties of dissolved substances can be collected by measuring their relaxation spectra in the solution \cite{brouwer2011spectra}. Roughly, a dissolved polymeric cluster can be modelled by a set of monomers (atoms) connected by elastic springs. If deformations of springs are small, the response of the molecule on external excitation is harmonic according to the Hooke's law. The relaxation modes are determined by the Laplacian matrix of the molecule. Measuring the response of the diluted solution of individual polymeric clusters on external excitation, on can see the signature of different eigenmodes in the spectral density as peaks at specific frequencies. In physical literature the spectrum of the adjacency matrix of a polymeric cluster typically is interpreted as the set of resonant frequencies, while the Laplacian spectrum provides the information about the typical relaxation times of the system.

Specifically, we consider a cluster of connected atoms as an $N$-vertex network (graph). Let us enumerate the atoms by the index $i=1 \ldots N$. The adjacency matrix $A = \{a_{ij}\}$ describes the topology (connectivity) of a cluster, it is symmetric ($a_{ij}=a_{ji}$) and its matrix elements, $a_{ij}$, take binary values, 0 and 1, such that diagonal elements vanish, i.e. $a_{ii}=0$. For off-diagonal elements, $i\ne j$, we set $a_{ij}=1$, if the vertices (atoms) $i$ and $j$ are connected, and $a_{ij}=0$ otherwise:
\be
a_{ij}= \begin{cases} 1 & \mbox{with probability $p$} \medskip \\ 0 & \mbox{with probability $1-p$} \end{cases}
\label{eq:01}
\ee
for $i\ne j$. The symmetric matrix $A$ is an adjacency matrix of a random Erd\H{o}s-R\'enyi graph $G$ without self-connections and double edges. The eigenvalues of $A$ are all real.

Spectrum and topology of \ERname{} graphs are controlled by the dependency of $p$ on $N$. Many results are known in cases when $p$ goes to zero slower than $1/N$. Meanwhile, there are many white spots in the case when $p = c / N$, where $c$ is a constant. In \cite{ks2003} Krivelevich and Sudakov proved that for $p\in(0,1)$ the typical largest eigenvalue is 
\begin{equation}
    \lev{} = (1+o(1))\max\{\sqrt{\dmax{}},Np\},
\end{equation}
where $\dmax{}$ is the maximal vertex degree in $G$. Intuition is the following. The largest eigenvalue of a star-graph is $\sqrt{\Delta}$, where $\Delta$ is the degree of the central node. Since the star with $\Delta=\dmax{}$ is a subgraph of $G$, the $\lev{}\geq{\sqrt{\dmax{}}}$. At the same time $\lev{} \geq \bar{d}$, the average degree in $G$. It turns out that there is a threshold between two cases: when $\lev{}$ is asymptotically determined by $\sqrt{\dmax{}}$ (``more sparse'' case) and by $\bar{d}$ (``denser'' case). In \cite{BenaychGeorges2019, BenaychGeorges2020,Alt2021} other eigenvalues in spectrum are analyzed for different regimes of $p$. Eigenvalues fluctuations are also an object of interest in literature. In recent work \cite{Bhattacharya2021} the lower and upper tail large deviations of $\lev{}$ are studied for $n^{o(1)-1}\ll p \ll \frac{1}{n}\sqrt{\frac{\ln N}{\ln \ln N}}$ (we discuss the fraction of logarithms later). In the preprint \cite{diaconu2022clts} $\lev{}$ is stated to have Gaussian fluctuations when $n^{\eps-1} \leq p \leq \frac{1}{2}$, $\eps\in(0,1)$. The cavity and replica methods of statistical mechanics are used in \cite{Susca2019, Susca2021, Kabashima2010, Kabashima2012} for studying the typical value of $\lev{}$ as well as the distribution of top eigenvector’s components in sparse graphs with bounded maximal degree.

In the present work we study the fluctuations of $\lev{}$ of \ERname{} graphs in the vicinity of the percolation point $p_c=1/N$. We start with the qualitative investigation of spectral boundaries and analytical derivation of the largest eigenvalue distribution for exponentially distributed linear chains. Then we conjecture that after appropriate choice of the normalization constant,
\begin{equation}
    \frac{\pi}{\arccos \frac{\levlin{}}{2}} \to \frac{\pi}{\arccos \frac{\lev{}}{\eignorm{}}},
\end{equation}
the proposed transformation of $\lev{}$ continues to be Gumbel distributed. We show that Gumbel distribution leads to the same scaling of finite-size corrections ($\ln^{-2} N$) established qualitatively in the vicinity of $p_c=1/N$. We also provide numerical simulations supporting the conjecture.

\section{Scaling estimates of spectral boundaries}

The ideas of works \cite{bowick}, applied to ensembles of Gaussian random matrices can be straightforwardly translated to the ensemble of random symmetric adjacency matrices $A$. Namely, we can estimate the finite size corrections to the eigenvalues which bound the main spectral zone in dense ($p=O(1)$) and sparse ($p=\frac{1}{N}$) ensembles of matrices $A$.

Let $\rho$ be the eigenvalue density of the ensemble of such matrices. For $p=O(1)$ in \eq{eq:01}, the spectral density, $\rho(\lambda)$, where $\lambda$ designates the eigenvalue of $A$, consists of the main zone in a form of a Wigner semicircle, $\rho_W(\lambda)$, typical for the Gaussian matrix ensembles, and one separated far-removed largest eigenvalue, $\lev{}$. The Wigner semicircle
\be
\rho_W(\lambda) = \frac{2}{\pi\levbound{}^2} \sqrt{\levbound{}^2-\lambda^2}
\label{eq:02}
\ee
bounds the main zone of the spectrum by the values $\pm \levbound{}$, where 
\be
\levbound{} = f(p)\sqrt{N}
\label{eq:boundary}
\ee
and $f(p)$ is some function of the connectivity, $p$.

To have a intuition about the typical behaviors of spectral densities $\rho(\lambda)$ in dense and sparse regimes, we have plotted in \fig{fig:01}a,b the function $\rho(\lambda)$ for $p=O(1)\approx 0.02$ (figure (a)) and for $p=\frac{1}{N}=0.0002$ (figure (b)) for ensembles of random adjacency matrices of size $N=5000$ with the Bernoulli distribution of matrix elements \eq{eq:01}. The plot in \fig{fig:01}c provides the spectral density of tridiagonal symmetric matrix with random distribution of off-diagonal elements: one has $a_{k,k+1}=1$ with the probability $p_{\rm lin}$ and $a_{k,k+1}=0$ with the probability $1-p_{\rm lin}$ (for all $1\le k\le N-1$, independent on $k$).

\begin{figure}[ht]
\centering
\includegraphics[width=0.9\textwidth]{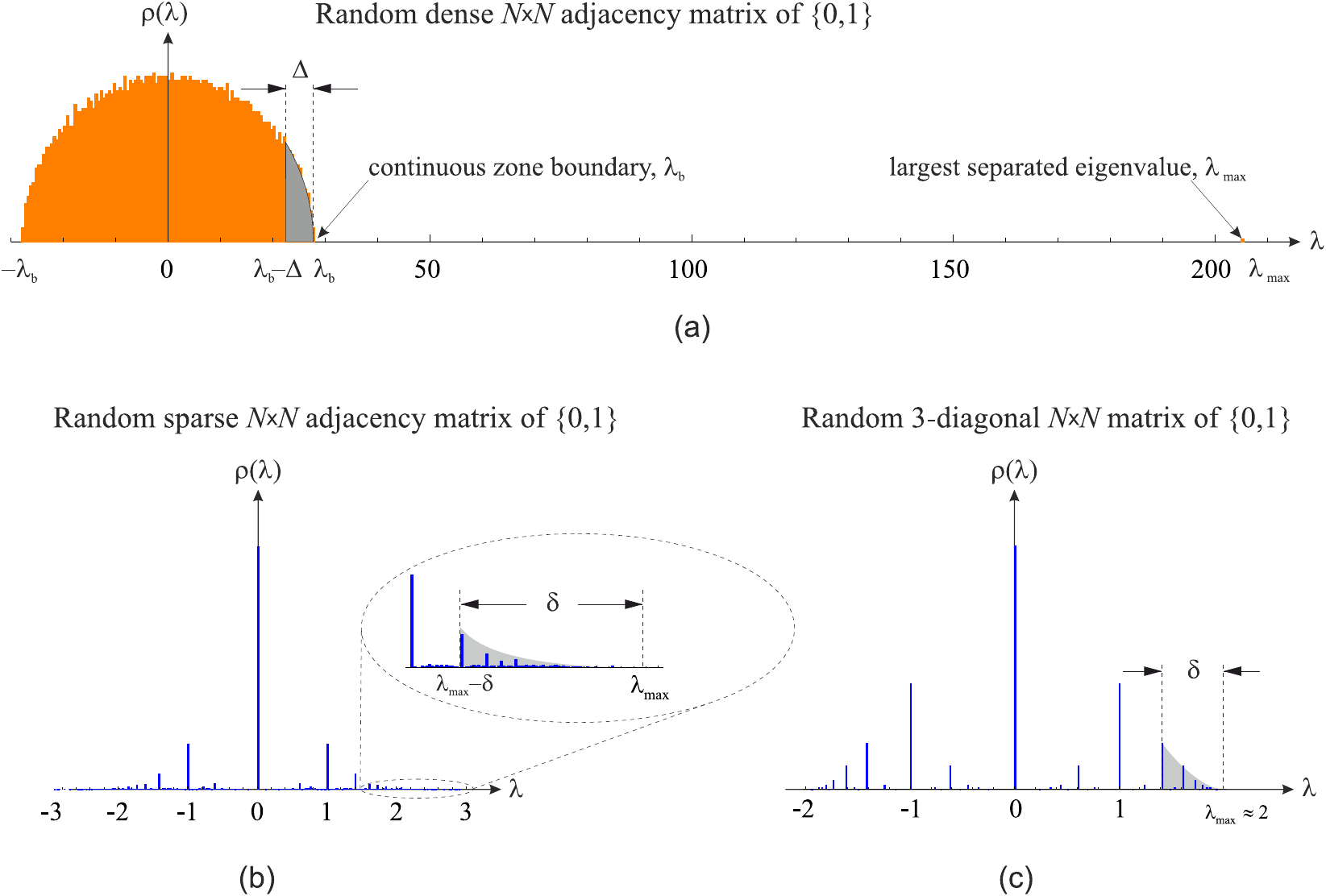}
\caption{Spectral densities of random $N\times N$ matrices in (a) dense, (b) sparse regimes, and (c) of a tridiagonal random operator with off-diagonal $\{0,1\}$ disorder.}
\label{fig:01}
\end{figure}

The behavior of the true maximal eigenvalue, $\lev{}$, follows from the Perron-Frobenius theorem, which states that $\lev{}$ of a positive matrix $A=\{a_{ij}\}$ satisfies the bilateral inequality
\be
\min_i \sum_{j=1}^N a_{ij} \le \lev{} \le \max_i \sum_{j=1}^N a_{ij}
\label{eq:03}
\ee
When $p\gg \ln N/N$ the graph is almost regular with high probability (Proposition 2.4.1 in \cite{Vershynin2018}) and $\lev{}$ is sandwiched between two variables, both of which $\approx \bar{d}$. Thus we arrive at the following expectation for $\lev{}$ at $N\gg 1$:
\be
\lev{} \approx p N
\label{eq:04}
\ee

The requested estimate of the finite-size correction, $\levbound{}(N)= f(p)\sqrt{N}$, to the main zone of the spectral boundary (see \fig{fig:01}a) in the dense regime is as follows. Suppose that the function $f(p)$ depends on $p$ only and is $N$-independent. Define the typical distance, $\Delta$, between adjacent eigenvalues in the vicinity of main zone boundary, $\levbound{}$. By definition the integral of $\rho(\lambda)$ over the interval $\left[\levbound{}-\Delta, \levbound{}\right]$ is the fraction of eigenvalues falling within this range, i.e.
\be
\label{eq:05}
\int\limits_{\levbound{}-\Delta}^{\levbound{}} \rho(\lambda)\, d\lambda \simeq \frac{1}{N}
\ee
Plugging \eq{eq:02} into \eq{eq:05} and taking into account that $\levbound{} = f(p)\sqrt{N}$, one arrives at the equation
\be
\frac{2}{\pi f^2(p)N} \int\limits_{f(p)\sqrt{N}-\Delta}^{f(p)\sqrt{N}} \sqrt{f^2(p)N-\lambda^2}\, d\lambda \simeq \frac{1}{N}.
\label{eq:06}
\ee
which provides an estimate of subleading scaling correction, $\Delta$ in the vicinity of $\lambda_{\rm max}$, valid at $N\gg 1$:
\be
\frac{4 \sqrt{2} \Delta^{3/2}}{3 \pi f^{3/2}(p) N^{3/4}} \simeq \frac{1}{N}; \qquad \Delta \simeq \frac{(3\pi)^{2/3}}{2^{5/3}} f(p) N^{-1/6}.
\label{eq:07}
\ee
Thus, the eigenvalue $\levbound{} \approx f(p)\sqrt{N}$ which bounds the continuous zone of the spectral density at large finite $N$ is defined with the uncertainty $\Delta \sim f(p) N^{-1/6}$, i.e.
\be 
\levbound{} \simeq f(p)\sqrt{N} \pm  f(p) N^{-1/6}
\label{eq:07a}
\ee

The same line of reasoning can be extended to estimate the uncertainty of the largest eigenvalue, $\lev{}$, of the sparse matrix ensemble at the percolation threshold, $p=p_c=1/N$. Note that now $\lev{}$ is not detached from other eigenvalues. It is known \cite{khorun, kirsch} that the spectral density, $\rho(\lambda)$, of an ensemble of sparse matrices near the spectral edge, $\lev{}$, has the singular behavior which manifests itself in the appearance of a``Lifshitz tail'',
\be
\rho(\lambda) \sim e^{-\frac{g(p)}{\sqrt{|\lev{}-\lambda|}}}
\label{eq:08}
\ee
where $g(p)$ is some function of the graph connectivity. Proceeding with \eq{eq:08} as with \eq{eq:02} and \eq{eq:05}, we get
\be
\int\limits_{\lambda_{\rm max}-\delta}^{\lambda_{\rm max}} e^{-\frac{g(p)}{\sqrt{|\lambda_{\rm max}-\lambda|}}}\, d\lambda \simeq \frac{1}{N}
\label{eq:09}
\ee
where $\lambda_{\rm max}$ is the maximal (boundary) eigenvalue of the spectrum in the sparse matrix ensemble. Defining the new variable $\delta=\lambda_{\rm max}-\lambda$ ($0<\lambda<\lambda_{\rm max}$) and performing the integration in \eq{eq:06}, we arrive at the equation for $\delta$ ($0<\delta\ll 1$):
\be
\left(\frac{2\delta^{3/2}}{g^3(p)}+O(\delta^{2})\right)e^{-\frac{g(p)}{\sqrt{\delta}}} \simeq \frac{1}{N}
\label{eq:10}
\ee
At $\delta\to 0$ the solution of \eq{eq:10} up to the leading term is:
\be
\delta \approx \left(\frac{g(p)}{\ln N}\right)^2
\label{eq:delta}
\ee
Thus, for $N\gg 1$ one arrives at the following finite size correction to the leading eigenvalue in sparse regime
\be
\lev{}(N)\big|_{N\gg 1} \simeq \lev{}(\infty) \pm \frac{g^2(p)}{\ln^2 N}
\label{eq:11}
\ee
More refined estimation of the asymptotic value $\lev{}(\infty)\equiv \lev{}$ in the vicinity of the percolation threshold is the subject of discussion provided in Section VI.

\section{\label{sec:linspec} Spectrum of linear chains}

Consider a symmetric $n \times n$ tridiagonal matrix\footnote{For tridiagonal random matrices and for the ensemble of linear subgraphs we use the lowercase index $n$, while the sparse matrices and the corresponding ensemble of graphs are denoted by the capital letter $N$.} $A_n$ composed of $n - 1$ Bernoulli variables $x_i$ ($i=1,..., n-1$):
\be
A_n = \left(\begin{array}{ccccc} 0 & x_1 & 0 & \cdots & 0 \smallskip \\  x_1 & 0 & x_2 & & \smallskip \\  0 & x_2 & 0 & & \smallskip \\ \vdots &  &  &  & \smallskip \\ & & & & x_{n-1} \smallskip \\ 0 & & & x_{n-1} & 0 \end{array} \right), \qquad \mbox{where}\; x_i=\left\{\begin{array}{ll} 1 & \mbox{with probability $p$} \medskip \\
0 & \mbox{with probability $1-p$} \end{array} \right.
\label{def:tria}
\ee
It is easy to see that the matrix $A_n$ has a block-diagonal structure of the following type
\be
A=\left(\begin{array}{ccc} \boxed{\begin{array}{ccc} 0\, & 1\, & 0\, \vspace{-0.05cm} \\ 1\, & 0\, & 1\, \vspace{-0.05cm} \\ 0\, & 1\, & 0\, \end{array}} \hspace{-0.2cm} & & \\ & \boxed{\begin{array}{cc} 0\, & 1\, \vspace{-0.05cm} \\ 1\, & 0\, \end{array}} \hspace{-0.2cm} & \\ & & \boxed{\begin{array}{cccc} 0\, & 1\, & 0\, & 0\, \vspace{-0.05cm} \\ 1\, & 0\, & 1\, & 0\, \vspace{-0.05cm} \\ 0\, & 1\, & 0\, & 1\, \vspace{-0.05cm} \\ 0\, & 0\, & 1\, & 0\, \end{array}} \end{array}\right)
\ee
where each block is a perfect tridiagonal matrix $B_j$ of a size, $n_j$, and blocks are uniquely defined by a sequence of zeros in $x_i$ ($i=1,...,n-1$). The spectrum of tridiagonal Toeplitz (diagonal-constant) matrices of size $n$ is given by a formula:
\be
\lambda(k,n) = a + 2 \sqrt{bc} \cos\frac{\pi k}{n+1}, \quad k = 1 \ldots n,
\ee
where $a,b,c$ are the values on main, upper and lower diagonals, respectively. Any linear chain has the adjacency matrix of exactly the same form, with $a = 0$ and $b = c = 1$. Thus, the eigenvalues of each $B_j$ have the form
\be
\lambda_{k,j} = 2 \cos \frac{\pi k}{n_j + 1}, \quad  \mbox{where $k = 1, \ldots, n_j$}
\label{eq:tri1}
\ee
We will be interested in the largest eigenvalue, which, as follows from \eq{eq:tri1}, corresponds to the minimum of the cosine argument, i.e.
\be
\lev{}(n_j) = 2 \cos \frac{\pi}{n_j + 1}
\label{eq:tri2}
\ee
Inverting \eq{eq:tri2} we can express the block size, $n_j$, as a function of the largest eigenvalue, $\lev{}(n_j)$:
\be
n_j = \frac{\pi}{\arccos\frac{\lev{}}{2}} - 1.
\ee

Since the determinant of the block matrix $\det (A - \lambda I)$ is the product of the determinants of blocks $\det(B_j - \lambda I)$, the spectrum of matrix $A$ is the union of the spectra of submatrices. Thus, the largest eigenvalue of $A$ is 
\be
\levtri{} = 2 \cos\frac{\pi}{\max_j n_j+1},
\label{eq:levtri}
\ee
where $\max_j n_j$ denotes the maximum length of consecutive set of ``1'' in $A$. Hence, the statistics of the variable
\be
n_{\max}=\frac{\pi}{\arccos\frac{\levtri{}}{2}} - 1
\ee
is governed by the distribution of $\max_j n_j$. In other words, knowing the extreme value statistics of the maximum linear length, $n_j$, we shall know the distribution of $\lev{}$, and vice versa.

\section{\label{sec:lsr} Longest success run}

In the block-diagonal matrix $A_n$ we have denoted by $\max n_j$ the size of the maximal block. Since blocks are formed between zeros (called ``failures'') in Bernoulli tests, $\max n_j$ is the maximum length of consecutive successes, which we call below as the``longest success run'' (LSR) in the sequence $x_1, x_2, \ldots, x_{n - 1}$. More precisely, let $\LSR(n,p)$ (or just $\LSR$) is the length of the maximum sequence of successes among $n$ Bernoulli trials, where $p$ is the probability of having ``1'' (see \eq{def:tria}).

Let us focus on $n \gg 1$. Since $x_j$ are independent, the number of $k$ consecutive successes has the geometric distribution  
\be
\prob\{\xi = k\} = p^k(1-p),
\label{eq:geomxi}
\ee
that can be replaced by the appropriate exponential distribution $\exp(\alpha)$ when $n \to \infty$:
\be
\prob\{\xi = k\} = \int\limits_{k}^{k + 1} \alpha e^{-\alpha t} \diff t = e^{-\alpha k} \left(1 - e^{-\alpha}\right).
\label{eq:exp}
\ee
Comparing \eq{eq:exp} with \eq{eq:geomxi} we get $\alpha = -\ln p$. The average number of zeros in the sequence is $\floor{n (1-p)}$. Therefore we can estimate $\LSR$ as $\LSR = \max\left(\eta_1, ..., \eta_{\floor{n (1-p)}}\right) - \tfrac{1}{2}$, where $\eta_j$ are i.i.d. random variables from $\exp(-\ln p)$ and $\tfrac{1}{2}$ is the continuity correction term.

The problem is now reduced to the following one: what is the limiting distribution of maximum of $\floor{n (1-p)}$ independent exponentially distributed variables? Using the Fisher-Tippett-Gnedenko theorem (see, for example \cite{trippett}) and considering the survival function, one arrives at the celebrated Gumbel distribution. Namely, if $\{X_j\}_{j\in\N}$ are the i.i.d. exponential random variables with the parameter $\alpha$ and $Y_n = \max_{1\leq j \leq n} X_j$, then
\be
\lim_{n\to\infty} \prob\left\{\frac{Y_n - \alpha^{-1}\ln n}{\alpha^{-1}} \leq z\right\} = e^{-e^{-z}}
\ee
This implies when $n \to \infty$
\be
\prob\left\{\LSR(n,p) < z \right\} = e^{-e^{-\frac{z-\mu}{\beta}}}
\ee
where
\be
\mu = \log_{1/p}n(1-p)-\tfrac{1}{2}, \quad \beta =-\frac{1}{\ln p}.
\label{eq:LSRGumbelparams}
\ee
Using the well-known properties of the Gumbel distribution, we get
\be
\EX{\LSR(n,p)} = \log_{1/p}n(1-p)-\frac{\gamma}{\ln p}-\tfrac{1}{2},
\label{eq:meanLSR}
\ee
where $\gamma\approx 0.5772$ is the Euler–Mascheroni constant.

\section{Large-$n$ corrections to the largest eigenvalue of a random tridiagonal matrix}
\label{sec:gumbellin}

Now we turn back to the tridiagonal matrix $A_n$ defined in \eq{def:tria} and its largest eigenvalue $\levtri{}$. From \eq{eq:levtri} we have
\be
\levtri{} = 2 \cos\frac{\pi}{(\LSR+1) + 1},
\ee
where $\LSR$ is the maximum length of consecutive ``1'' on the subdiagonal of $A_n$ (the size of the corresponding block is $n+1$). Having the Gumbel distribution for $\LSR$, we can immediately derive the related distribution for $\lev{}$ of the random tridiagonal matrix:
\be
\prob\left\{\frac{\pi}{\arccos\frac{\levtri{}}{2}} < x\right\} = \prob\left\{\LSR \leq x-2\right\} = e^{-e^{-\frac{x-2-\mu}{\beta}}}.
\label{eq:tridiaglevCDF}
\ee
Taking the logarithm twice and substituting \eq{eq:LSRGumbelparams} into \eq{eq:tridiaglevCDF}, we get the linear function of $x$:
\be
\ln(-\ln F) = \left(x - \log_{1/p}n(1-p)-\tfrac{3}{2}\right)\ln p,
\label{eq:tridiagrescaledGumbel}
\ee
where $F=F(x)$ is the cumulative probability function (CDF) in the LHS of \eq{eq:tridiaglevCDF}. Now, using \eq{eq:meanLSR} we can find the mean value
\be
\EX{\left\{\frac{\pi}{\arccos\frac{\levtri{}}{2}}\right\}} = \log_{1/p}n(1-p)-\frac{\gamma}{\ln p}+\tfrac{3}{2}.
\label{eq:tridiagmeanx}
\ee

Equations \eq{eq:tridiagrescaledGumbel} and \eq{eq:tridiagmeanx} are in excellent agreement with the results of numeric simulations. Corresponding plots are shown in \fig{fig:meantri}.
\begin{figure}[ht]
    \centering
    \subfloat[Sampled CDF of $\pi/\arccos(\lev{}/2)$ for $n=10^6$ and $p=0.5$, number of bins is $1000$]{ 
        \includegraphics[width=0.46\columnwidth]{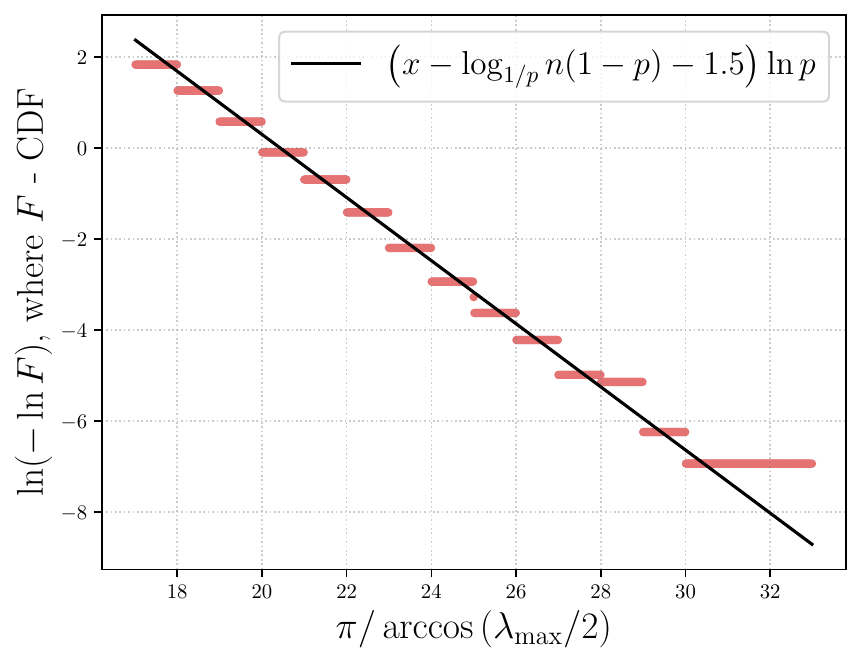}
    }\hfill%
    \subfloat[Mean $\pi/\arccos(\lev{}/2)$ sampled for different $n$ and $p$]{
        \includegraphics[width=0.48\columnwidth]{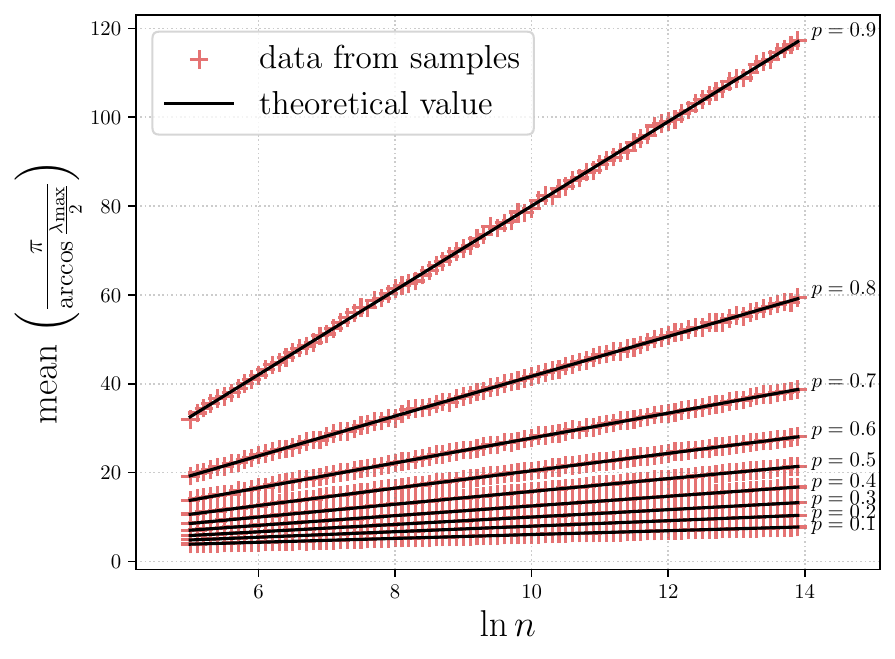}
    }
    \caption{\label{fig:meantri} Comparison of analytic and numerical results for equations \eq{eq:tridiagrescaledGumbel} and \eq{eq:tridiagmeanx} based on $10^5$ tridiagonal samples for each pair of $n$ and $p$.}
\end{figure}

It should be pointed out that the asymptotic expression \eq{eq:tridiagmeanx} is consistent with the one obtained in Introduction via the naive estimation of $\EX\left(\lev{}\right)$ based on the analysis of eigenstates in the Lifshitz tail of the density $\rho(\lambda)$ near the spectral edge in the sparse matrix ensemble -- see \eq{eq:11}. Expanding $\pi/\arccos(\lev{}/2)$ near the spectral edge of linear chains, $\lev{}=2$, we find:
\be
\frac{\pi}{\arccos\frac{2-\delta}{2}}\bigg|_{\delta\to 0} \approx \frac{\pi}{\sqrt{\delta}}-\frac{\pi \sqrt{\delta}}{24} + O\left(\delta^{3/2}\right) 
\label{eq:asymp}
\ee
Substituting \eq{eq:asymp} into \eq{eq:tridiagmeanx}, we get for $\delta$ the following expression
\be
\delta = \left.\frac{\pi^2}{\left(\log_{1/p}n(1-p) - \frac{\gamma}{\ln p} + \frac{3}{2}\right)^2}\right|_{n\gg 1} \approx \frac{\pi^2}{\log^2_{1/p}n} = \frac{\pi^2\ln^2 p}{\ln^2 n}
\label{eq:estim}
\ee
Comparing \eq{eq:estim} and \eq{eq:11} one sees that both expressions have the same dependence on $n$. The non-rigorous nature of derivation in Introduction does not permit to rely on coefficients in \eq{eq:11}, while \eq{eq:estim} provides correct answer (confirmed numerically) in the large-$n$ limit.

\section{Spectra of sparse \ERname{} graphs}

Let $G \sim \fancyG(N, p)$ be a random \ERname{} graph with $N$ vertices and the probability $p$ of an edge formation. Here is a brief recap of how the structure of $G$ evolves with the increase of $p$. At $p < p_c = \frac{1}{N}$ linear chains statistically suppress branching graphs, so the spectrum of $G$ is entirely determined by linear chains and lies within the interval $[-2, 2]$ -- see \fig{fig:samples}a plotted at $p = \frac{0.5}{N}$. As we approach the percolation point, subgraphs with $z>2$ branchings start to contribute, and the giant component increases in size becoming of order of $N^{2/3}$ at the percolation point. The corresponding regime is depicted in \fig{fig:samples}b at $p = \frac{1.5}{N}$ i.e. slightly above the percolation point. As $p$ increases further, the giant component continues to grow, crowding out all other subgraphs. At the point $p_c^* = \frac{\ln N}{N}$ a cloud of short linear chains and isolated vertices floats around, and at $p > (1+\eps) p_c^*$ for any $\eps>0$ the graph $G$ almost surely becomes connected in the thermodynamic regime. The corresponding phases are illustrated in \fig{fig:samples}c where $p \approx \frac{\ln N}{N}$.

\begin{figure}
    \centering
    \includegraphics[width=0.32\linewidth]{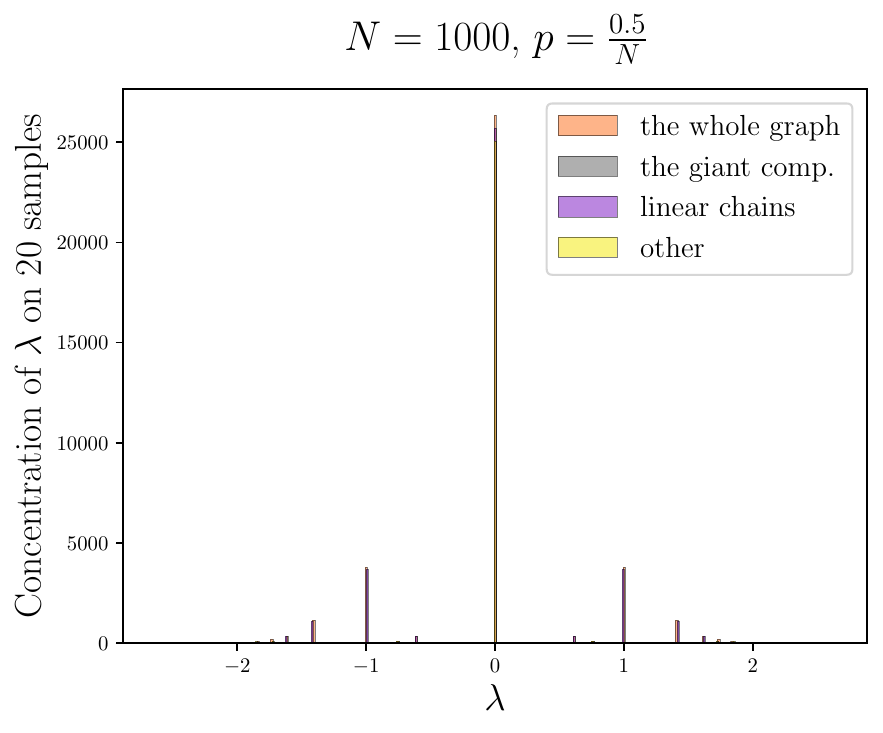}\hfill
    \includegraphics[width=0.32\linewidth]{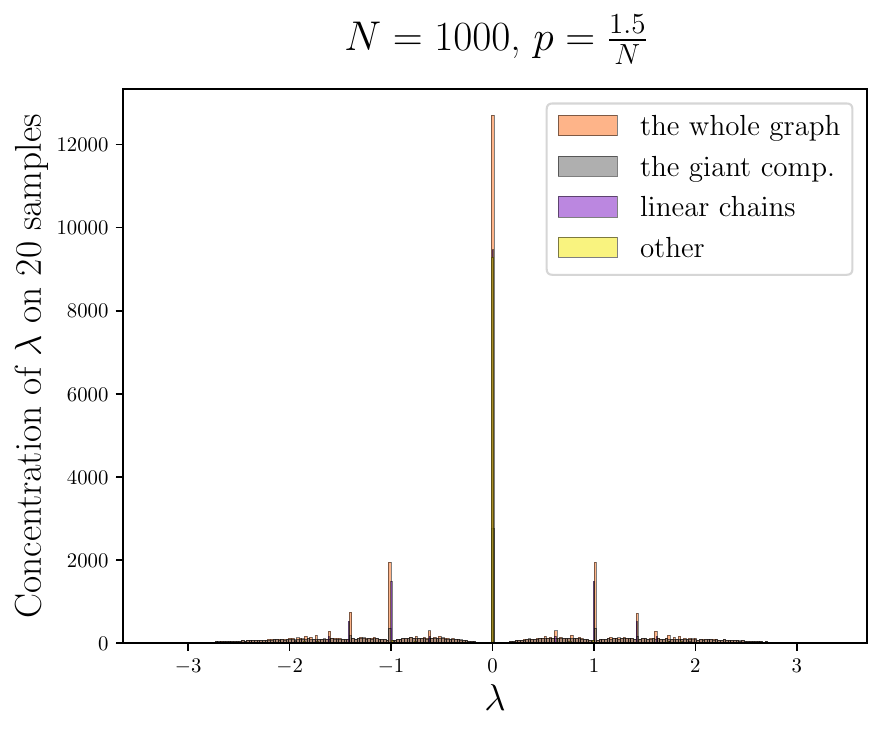}\hfill
    \includegraphics[width=0.32\linewidth]{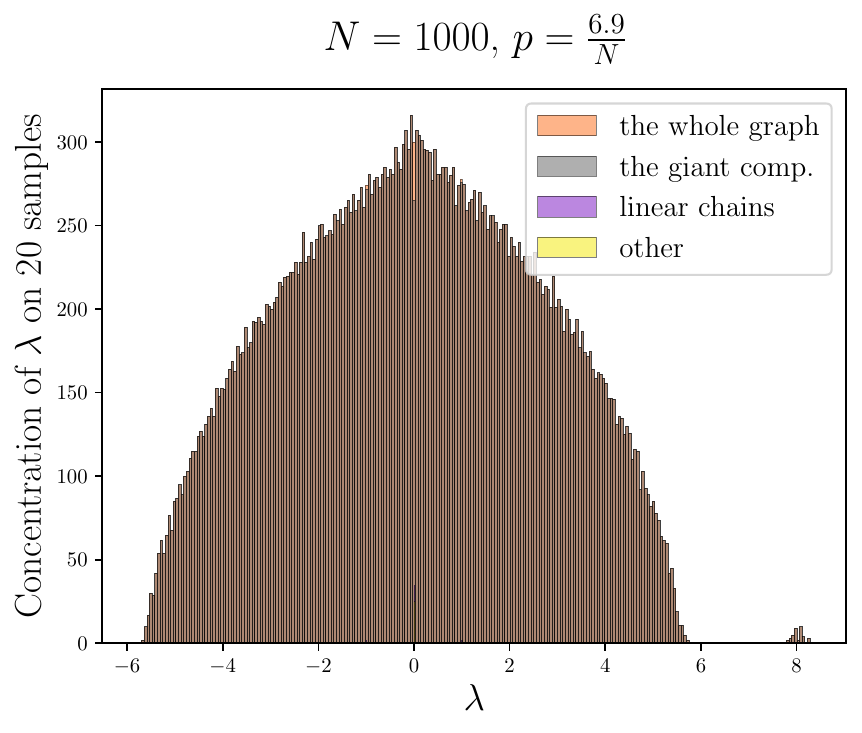}
\caption{Spectrum of ${\cal G}(N, p)$ grouped by types of subgraphs.}
\label{fig:samples} 
\end{figure}

\subsection{Contribution of linear subgraphs}

In the case when linear subgraphs dominate (see \ref{sec:linspec}) the largest eigenvalue of the random adjacency is determined by the maximal length of a chain. In \cite{Avetisov2015} it has been shown that in the vicinity of the percolation threshold $p_c=1/N$ linear subchains provide the dominant contribution to ${\cal G}$ and are exponentially distributed with the law $P(L) \sim e^{-L}$, where $L$ is the subchain length. As we know from Section \ref{sec:gumbellin}, that leads to the Gumbel distribution of the variable
\be
\varnlin{}=\frac{\pi}{\arccos\frac{\levlin{}}{2}},
\label{eq:quanlinER}
\ee
where $\levlin{}$ is the largest eigenvalue for linear chains separated from all other components in the \ERname{} graph $G$. The spectrum of any graph is the union of spectra of its connectivity components. Since the largest eigenvalue $\lev{}$ in $G$ can be determined by a non-linear component, the following natural question emerges: could the distribution of $\lev{}$ still be estimated if the graph topology is known?

There are two upper bounds for $\lev{}$ depending on the maximum vertex degree (or ``branching'') $\dmax{}$. One can be applied to any graph
\be
\lev{} \leq \dmax{} \quad \text{(any graph)},
\label{eq:anytypebound}
\ee
whereas another one is about trees (\cite{Stevanovic2003}),
\be
\lev{} \leq 2\sqrt{\dmax{} - 1} \quad \text{(tree)}.
\label{eq:treetypebound}
\ee
Both inequalities \eq{eq:anytypebound} and \eq{eq:treetypebound} are sharp: the first one becomes equality for a complete graph and the second one -- for an infinite regular tree. For linear chains upper bounds coincide, since $\dmax{} =  2\sqrt{z-1}$ for $z=2$. Now we are in position to formulate the main conjecture: 

\noindent {\bf Conjecture.} {\it Taking into account that the normalization constant $2$ in the denominator of \eq{eq:quanlinER} is the spectral boundary of ensemble of linear graphs, to extend our consideration beyond the ensemble of linear graphs, let us replace $2$ in the denominator of \eq{eq:quanlinER} by the spectral boundary $\eignorm{}$ for ensemble of sparse graphs generated at some value $p$:
\be
n=\frac{\pi}{\arccos\frac{\lev{}}{\eignorm{}}}
\label{eq:CquanER}.
\ee
}

Gumbel distribution traced for $\varnlin{}$ in \eq{eq:quanlinER} for linear subgraphs motivates to suppose its validity in \eq{eq:CquanER} even beyond the linear case. Since components in the \ERname{} model are not obliged to be trees even at the percolation point (the probability $\prob\{\text{no cycle in}\ G\}$ tends to $0$ when $N \to \infty$ as it is proved in \cite{Erdos1960}), it is better to use the first inequality \eq{eq:anytypebound} to find an appropriate normalization constant $\eignorm{}$ in \eq{eq:CquanER}. Thus, we end up with the question of finding the best estimate for $\dmax{}$ in ensemble of sparse graphs above the percolation threshold.

\subsection{Maximum degree bounds}

Every vertex degree is a sum of elements in the corresponding row of the adjacency matrix. There are many inequalities estimating tails of sum of Bernoulli variables. However, Chernoff inequality (\cite{Vershynin2018}) turns out to be an effective tool, when the mean value of sum is $O(1)$.
\begin{proposition}
    \label{prop:upperbounddmax}
    Consider a random graph $G \sim \fancyG(N, 1/N)$. Then for any $\eps \in (0,1)$ and $\delta>0$ there is a positive constant $M$ such that for any $N > M$ the following inequality is valid: 
    \begin{equation*}
       \prob\left\{\exists i \in G: d_i \geq (1+\delta)\frac{\ln N}{\ln \ln N}\right\} \leq \eps,
    \end{equation*}
    where $d_i$ is the degree of node $i$. In other words, for any $N > M$ with the probability at least $1 - \eps$ all vertices in $G$ have degree less than $(1+\delta)\ln N/\ln\ln N$.
\end{proposition}

\begin{proof}
Let us use denote by $i$ vertices in $G$ and by $d_i$ the corresponding vertex degree. We begin with the following bound:
\be
\prob\left\{\exists i: d_i \geq t\right\} \leq N \prob\left\{d_i \geq t\right\},
\ee
where $t = \frac{(1+\delta)\ln N}{\ln\ln N}$\smallskip. Certainly $\mu \equiv \EX{d_i} = \frac{N-1}{N} \leq 1$ and using the Chernoff inequality we get
\be
\prob\left\{\exists i: d_i \geq t\right\} \leq N e^{-\mu}\left(\frac{e\mu}{t}\right)^t \leq N \left(\frac{e}{t}\right)^t \leq \exp\left(\ln N + t - t \ln t\right) \leq \exp[b(N,\delta)].
\label{eq:prop1ineq}
\ee
We are about to show that the last expression tends to $0$, since the expression in the brackets tends to $-\infty$. The explicit form of $\ln t$ is 
\be
\ln t = \ln (1+\delta) + \ln\ln N - \ln\ln\ln N.
\ee
Substituting this expression in \eq{eq:prop1ineq} and rearranging the terms, we get the exponent
\be
b(N,\delta)=\frac{\ln N}{\ln\ln N}\bigl[-\delta\ln\ln N + (1+\delta)(\ln\ln\ln N + 1-\ln(1+\delta))\bigr].
\label{eq:prop1exponent}
\ee
Fixing any positive $\delta$, the whole expression in \eq{eq:prop1exponent} tends to $-\infty$, since the leading term inside brackets is $-\delta\ln\ln N$. Returning to \eq{eq:prop1ineq} we see that the upper bound tends to $0$. Namely, starting with some positive $M$ the probability of having a node with the degree $\geq t$ will be less than $\eps$. In other words, all degrees in this expression are less than $t$ with probability bigger than $1-\eps$.
\end{proof}
Surely, if all degrees are less than some value, it does not necessarily mean that there are some degrees close to that value. 
The upper bound \eq{eq:prop1ineq} with the explicit expression for the exponent \eq{eq:prop1exponent} reveals the interplay between $N, \delta$, and desired probability, $1-\eps$. This becomes important when one generates any \emph{finite} collection of graphs $G\sim\fancyG(N, \frac{1}{N})$ with \emph{finite} $N$. We address this question in next Section.

\subsection{Gumbel statistics related to $\lev{}$ in sparse graphs}

Generally, if $E$ is an event happen with the probability, say, $1/1000$, then in practical computations we can expect detecting about one such event, $E$, among $1000$ independent samples. This implies that if we work with $N$ and $\delta$ such that the upper bound \eq{eq:prop1ineq} is $\ll 1/1000$, we can neglect the occurrence of $E$ among $1000$ independent samples. The Table \ref{tab:ubnumerical} provides upper bounds for different values of $N$ and $\delta$. Decreasing $\delta$ by $0.5$ leads to increasing the upper bound by several orders of magnitude. So, to reduce the upper bound one need to increase $N$ significantly. For example, one can generate about 1000 graphs $G$ of sizes $N \times N$ (say, $N>10^5$) and with a high probability there will be no vertices with the degree bigger than $\frac{3\ln N}{\ln\ln N}$ in the whole collection.

\begin{table}[ht]
\begin{tabular}{ccc}
\hline \hline
$N$ & \hspace{1cm} $\delta$  & \hspace{1cm} $\exp(b(N,\delta))$ \\ \hline \hline
$10^5$ & \hspace{1cm} $2$  & \hspace{1cm} $\approx 8 \times 10^{-6}$ \\ \hline
$10^5$ & \hspace{1cm} $1.5$ & \hspace{1cm} $\approx 4 \times 10^{-3}$ \\ \hline \hline
$10^6$ & \hspace{1cm} $2$   & \hspace{1cm} $\approx 9 \times 10^{-7}$ \\ \hline
$10^6$ & \hspace{1cm} $1.5$ & \hspace{1cm} $\approx 1 \times 10^{-3}$ \\ \hline
\end{tabular}
\caption{Numerical values of \eq{eq:prop1ineq}. Graph $G\sim\fancyG(N, 1/N)$, $\prob\left\{\exists i \in G: d_i \geq (1+\delta)\ln N / \ln\ln N \right\} < \exp(b(N,\delta))$ (see Proposition \ref{prop:upperbounddmax}).}
\label{tab:ubnumerical} 
\end{table}

Depending on the number of samples one can choose the appropriate value of $\delta$ providing the upper bound for $\dmax{}$ with a high probability. Since $\delta$ is of order $1$ and there is a probabilistic gap between $\lev{}$ and $\dmax{}$ depending on exact topology of each sampled graph, one can assume that $\frac{\ln N}{\ln\ln N}$ is a good bound in the large-$N$ limit. That motivates us to update the normalization constant in \eq{eq:CquanER} with the new norm and get a new quantity:
\begin{equation}
    x = \frac{\pi}{\arccos\left(\lev{}/\frac{\ln N}{\ln\ln N}\right)}
    \label{eq:lnlnquanER}.
\end{equation}

Below we provide numerical arguments in support of the hypothesis that not only the largest eigenvalue $\levlin{}$ in the ensemble of {\it linear chains}, but also the largest eigenvalue $\lev{}$ in the ensemble of \emph{sparse graphs} shares the Gumbel distribution at least in the vicinity of the percolation point, $p_c = 1/N$. We have seen in \eq{eq:tridiaglevCDF} that the Gumbel statistics implies CDF of the form
\begin{equation}
    F(x) = e^{-e^{-\frac{x-\mu}{\beta}}}.
    \label{eq:gumbel}
\end{equation}
Our check depicted in \fig{fig:threeNexperiment} is based on the numerical verification of linearity in doubly logarithmic coordinates of the cumulative distribution function $F(x)$. 
\begin{figure}[ht]
    \centering
    \includegraphics[width=0.33\linewidth]{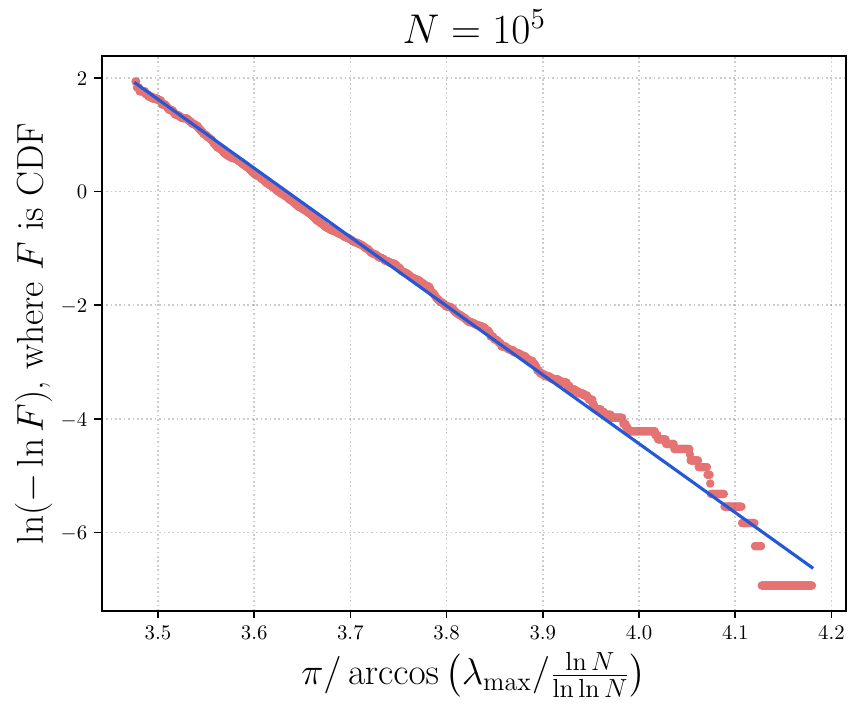}\hfill
    \includegraphics[width=0.33\linewidth]{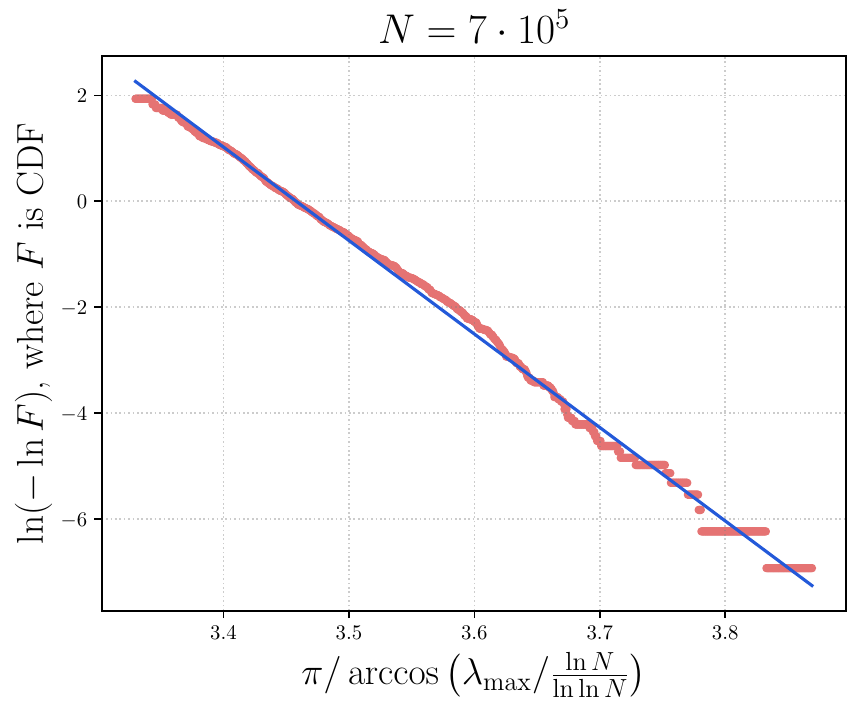}\hfill
    \includegraphics[width=0.33\linewidth]{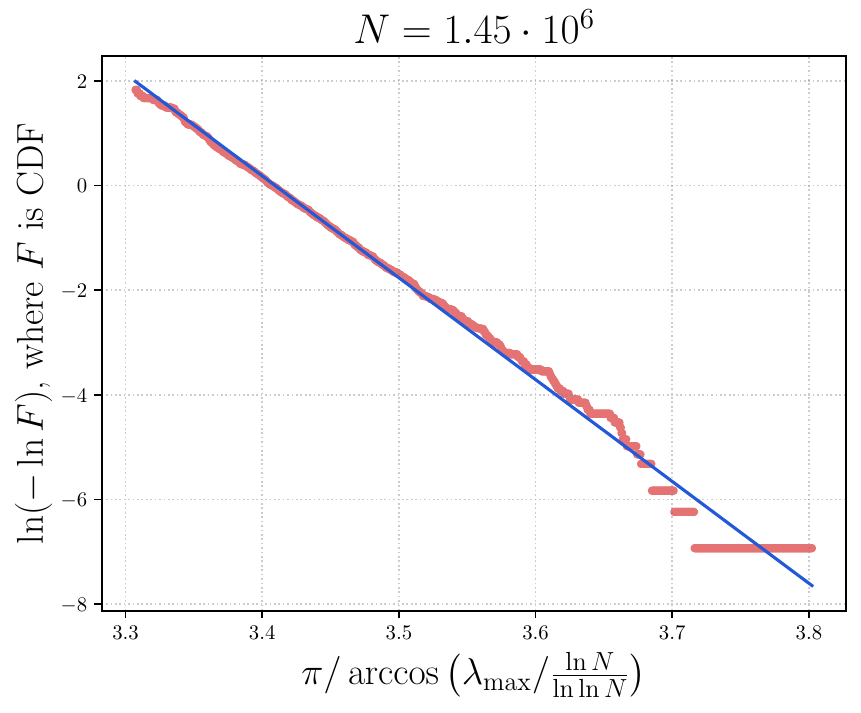}
    \caption{Numerical cumulative distribution function (CDF) $F$ based on $1024$ randomly sampled \ERname{} graphs with edge probability equal $1/N$. Number of bins is $1000$. Numerical value of $F(e)$ at the bin edge $e$ defined as number of sampled $x$ values $<{e}$ divided by the number of samples. Bin edges where $F$ is $\leq \eps_F=10^{-9}$ or $\geq 1-\eps_F$ were excluded.}
    \label{fig:threeNexperiment}
\end{figure}

For each $N$ we generate $1024$ graph samples with the edge probability $1/N$ and calculate $1024$ values of $x$ from \eq{eq:lnlnquanER}. Then we divide the region $[\min x - \eps_x, \max x + \eps_x]$ into $1000$ bins, where $\eps_x = 10^{-9}$ (this number is much less than any value of $x$) and find $\min/\max$ among sampled $x$ (for each $N$ individually). Now we have a collection of tuples $(e, F(e))$, where $e$ is the bin edge, $F({\rm e})$ is the number of sampled $x$ values less than $e$ divided by the number of samples (this is the definition of the numerical cumulative distribution function (CDF)). Since our CDF at first and last bins can be $0$ and $1$, the value of $\ln (-\ln F)$ is not defined at that points and we consider only bin edges where numerical CDF is $>\eps_F$ and $<1-\eps_F$ ($\eps_F = 10^{-9}$). Noise at tails in \fig{fig:threeNexperiment} is a natural consequence of having rare ``superlarge'' or ``supersmall'' sampled $x$ where numerical values of $F$ do not follow the main trend because of the lack of samples around these $x$.


\subsection{Empirical choice of normalization constant}

Here we suggest the numerical procedure which permits to choose the desired norm $\eignorm{}$ in \eq{eq:CquanER}. Let us scan all possible $\eignorm{}$, calculate residuals of the linear fit and choose the smallest one among them. This prescription provides the value corresponding of the correct norm for our particular collection of samples. 

Let us first test how this algorithm works for the tridiagonal case (i.e. for linear subgraphs). For small $p$ fluctuations of $\max$ length of the continuous sequence of ones are of the order of its length. Meanwhile for large $p$ the $\max$ length becomes of order of $N$ and different chains of ones start to affect the statistics of each other. Recall that we established the Gumbel distribution \eq{eq:tridiaglevCDF} when $N \to \infty$ at fixed $p$. So, for large $p$ we need to increase $N$ to be able to simulate the thermodynamic limit. Since it only affects our computational resources we set for simplicity $p = 0.5$. Below we provide the results of numerical investigation of $\levtri{}$ norm for linear chains. The corresponding plots are shown in \fig{fig:bestnormtridiag}. Normalization constant $\eignorm{}$ in \fig{fig:bestnormtridiag} is changing within the interval $[1.99,2.10]$ with the step $10^{-5}$. Since any $\eignorm{}$ must be greater than the sampled $\max\levtri{}$, we skip values of $\eignorm{}$ less than $(\max\levtri{}-\eps_{\lambda})$, where $\eps_{\lambda}=10^{-9}$. For each $N$ and $\eignorm{}$ we repeat the algorithm described at the end of previous Section. We observe a narrow fall of the residual until local minimum is reached. The local minimum is followed by a steady sloping growth. The best norm fluctuates around $2$, however not always it is exactly equal to 2. It happens because the sampled finite statistics is not sufficient for finding the true normalization constant.

\begin{figure}[ht]
    \centering
    \includegraphics[width=0.33\linewidth]{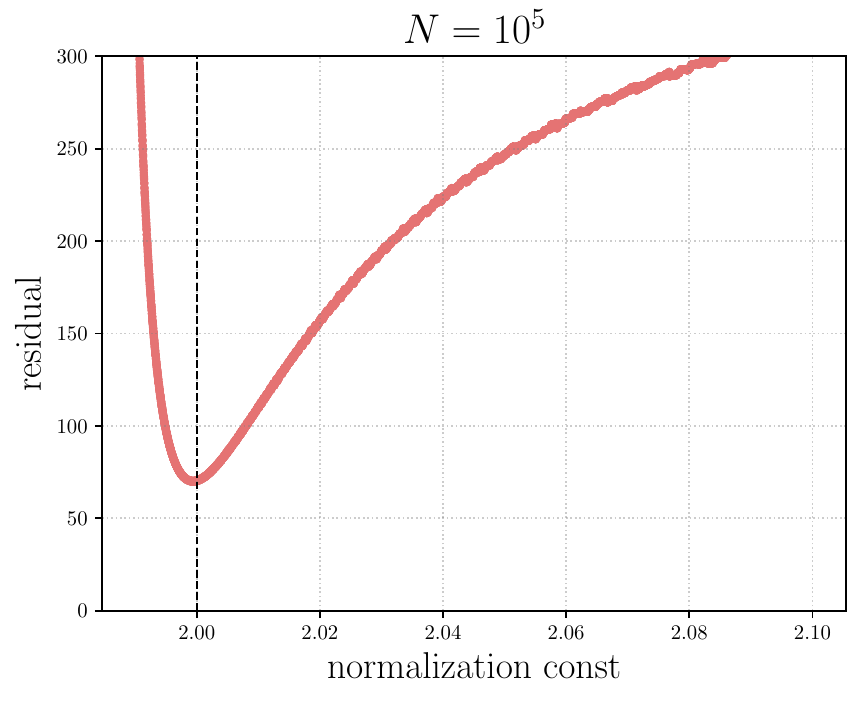}\hfill
    \includegraphics[width=0.33\linewidth]{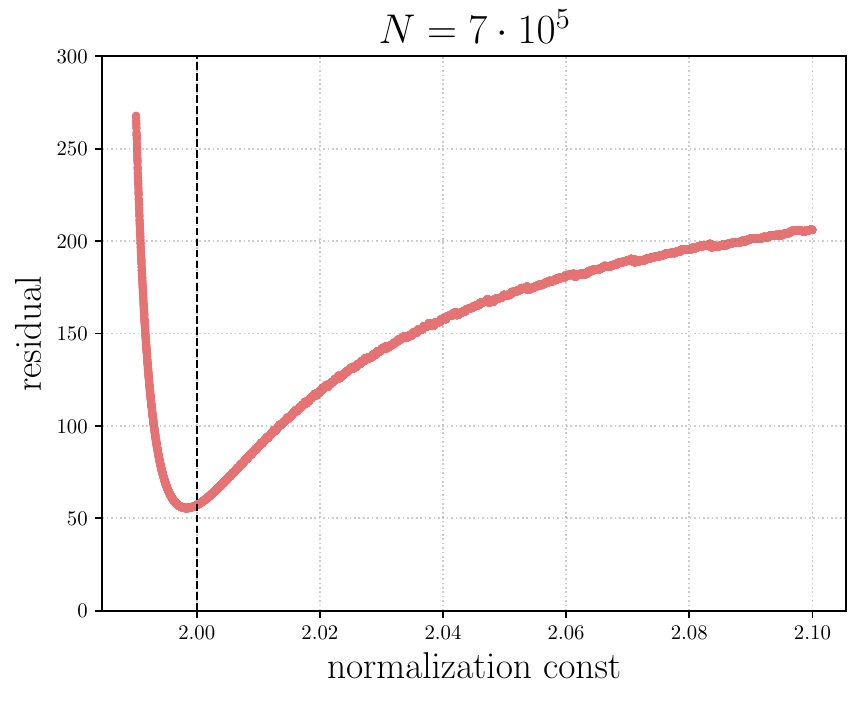}\hfill
    \includegraphics[width=0.33\linewidth]{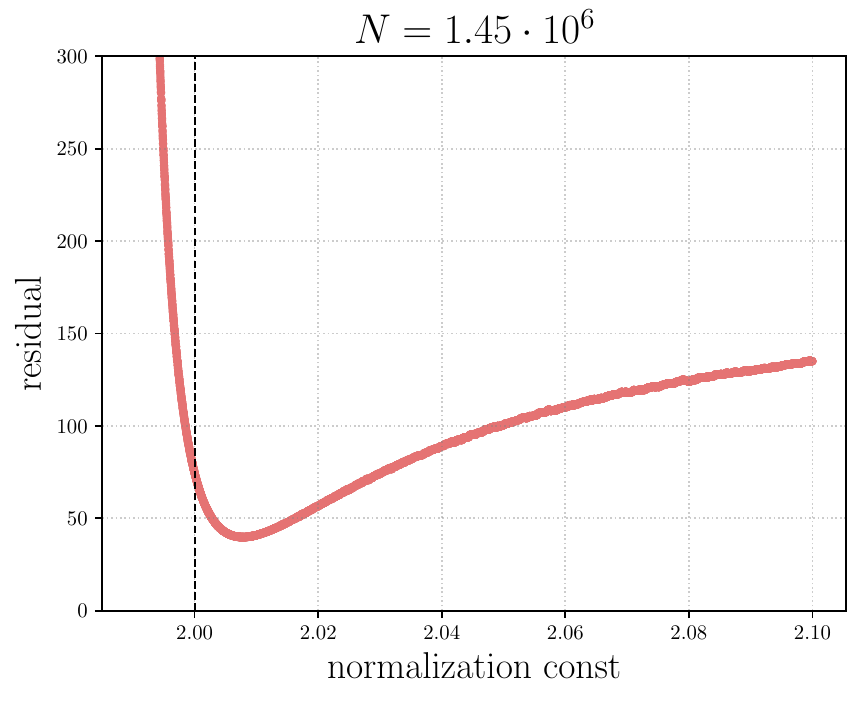}
    \caption{\label{fig:bestnormtridiag}
        Linear fit residuals of $\ln(-\ln F)$, where $F$ is the numerical CDF of sampled $\pi/\arccos{(\levtri{}/\eignorm{})}$ in the tridiagonal case for different normalization constant. $1024$ samples for each $N$.
    }
\end{figure}

Turning to sparse graph simulations one should take into account two technical circumstances. First, for a fixed $N$ the configurational space of sparse graph ensemble is much bigger than that of linear chains. Second, it is much heavier computational task to calculate the largest eigenvalue of a sparse matrix than of a tridiagonal one. Figure \ref{fig:bestnormsparse} repeats for sparse matrices the construction shown in \fig{fig:bestnormtridiag} for tridiagonal case.

\begin{figure*}[ht]
    \centering
    \includegraphics[width=0.33\linewidth]{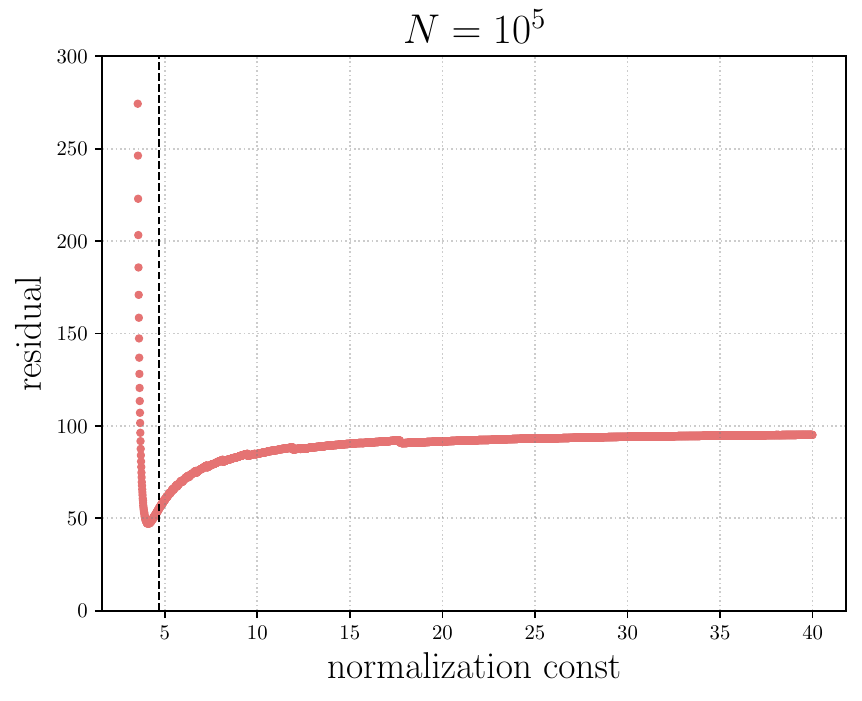}\hfill
    \includegraphics[width=0.33\linewidth]{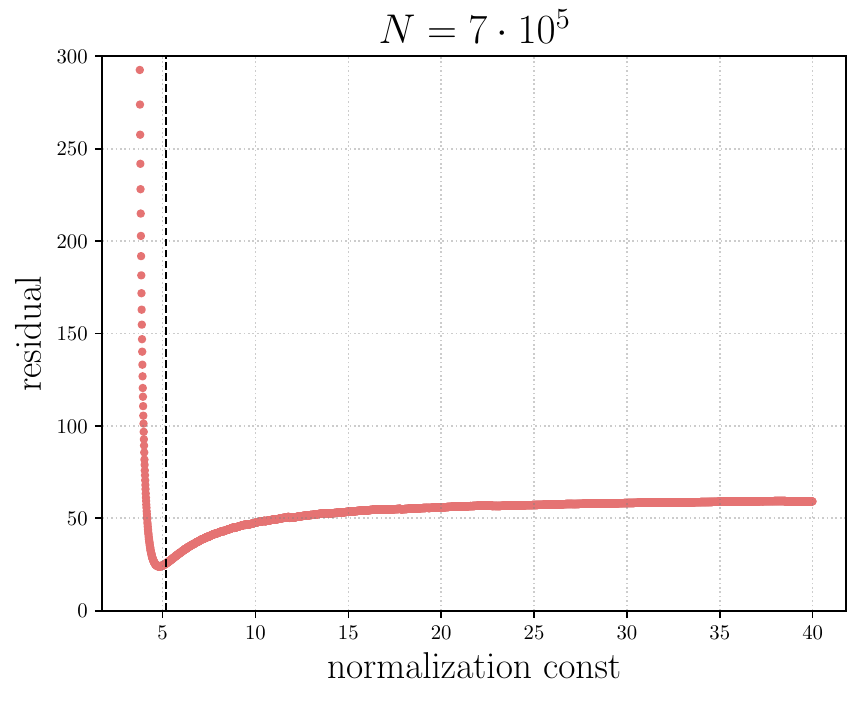}\hfill
    \includegraphics[width=0.33\linewidth]{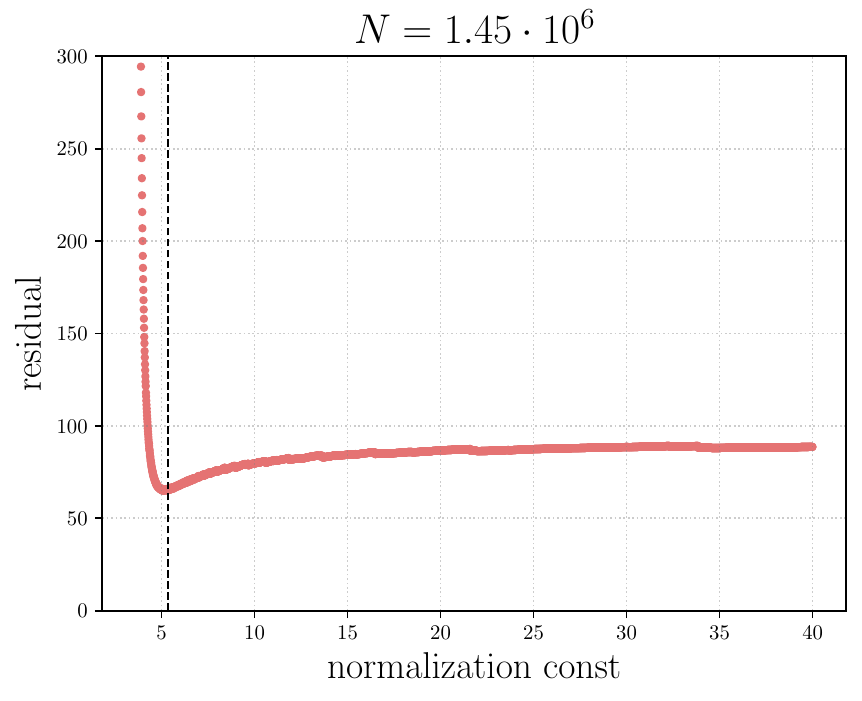}
    \caption{\label{fig:bestnormsparse}
        Linear fit for residuals of $\ln(-\ln F)$, where $F$ is the numerical CDF of sampled $\pi/\arccos{(\lev{}/\eignorm{})}$ in the sparse case for different normalization constants. Dashed vertical lines are at $\eignorm{}=\ln N/\ln\ln N$. We have generated $1024$ random samples from $\fancyG(N,1/N)$ for each $N$.
    }
\end{figure*} 

We iterate $\eignorm{}$ with a step $10^{-2}$ over the interval $[3, 40]$ skipping values that are less than $(\max\levtri{}-\eps_{\lambda})$. Residuals in \fig{fig:bestnormsparse} follow the same pattern which we had for linear subgraphs: sharp narrow decay till the minimum fluctuating around $\eignorm{} = \ln N/\ln \ln N$ followed by a steady growth. The choice of $\eignorm{}$ as $\ln N/\ln \ln N$, but not as $\sqrt{\ln N/\ln \ln N}$ as suggested in \cite{ks2003} is discussed below.

Repeating many times random sampling, sometimes we do not see the minimum around $\eignorm{}$, where it typically occurs in \fig{fig:bestnormsparse}. This happens because of insufficient number of samples. Generating additional set of graphs, the minimum emerges and becomes more profound (see \fig{fig:moresamplessparse}). The wide plateau of residuals in \fig{fig:moresamplessparse} means that in this region the results are practically insensitive to the normalization constant $\eignorm{}$, signaling that values of $\eignorm{}$ of order of $\ln N/\ln \ln N$ are as good as the ones of order of $\sqrt{\ln N/\ln \ln N}$.

\begin{figure}[b]
    \centering
    \includegraphics[width=0.43\columnwidth]{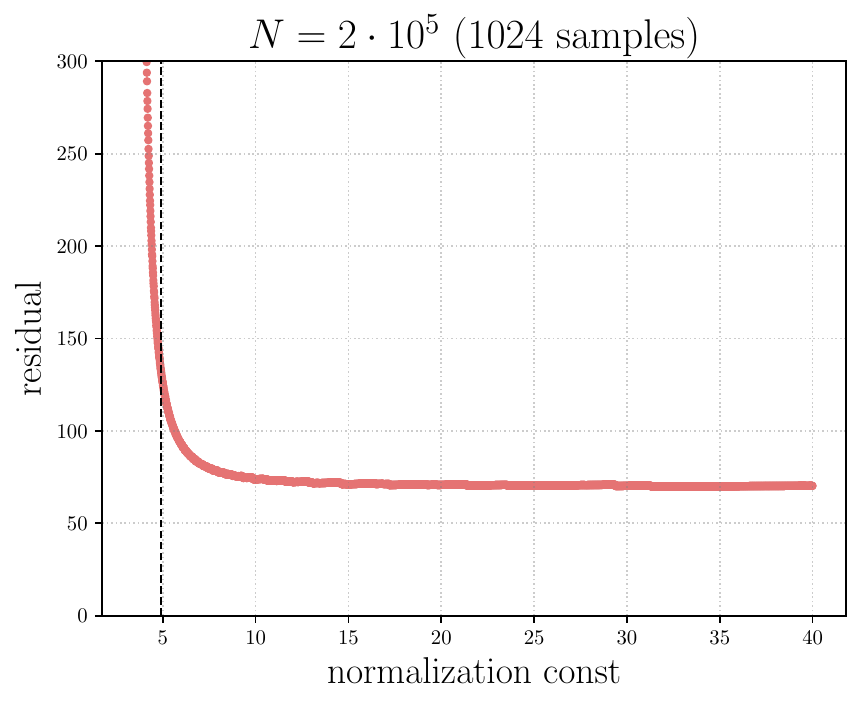}\hfill
    \includegraphics[width=0.43\columnwidth]{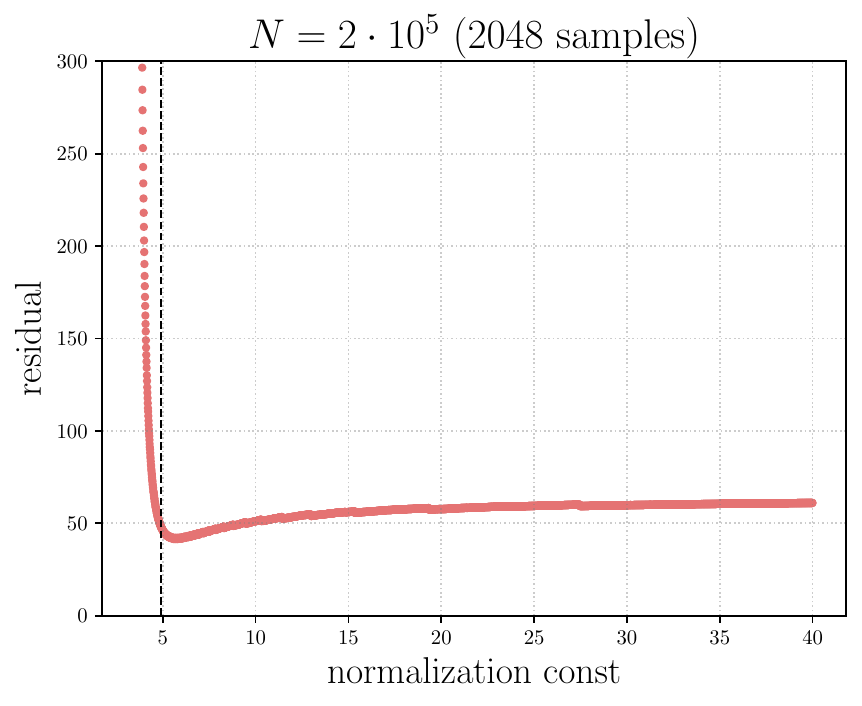}
    \caption{\label{fig:moresamplessparse}
    Increasing number of samples leads to a pronounced local minimum corresponding the best normalization constant $\eignorm{}$ in \eq{eq:CquanER}. Dashed vertical lines are at $\eignorm{}=\ln N/\ln\ln N$. When the sampled statistics is not enough for minimum to occur, than $\ln N/\ln\ln N$ happens to be in a region, where the plateau starts.}
\end{figure}

It is worth mentioning that the estimate $\lev{}(N) = (1+o(1))\sqrt{\ln N/\ln\ln N}$ derived in \cite{ks2003} is the best known estimate of the asymptotics of the largest eigenvalue in a sparse graph ensemble. However in our numerical simulations this estimate cannot be used as a normalization constant. The reason is as follows: for any {\it finite} $N$ and {\it finite} set of samples many graphs have $\lev{} > \sqrt{\ln N/\ln\ln N}$ which means that we can not extract $N$ from \eq{eq:CquanER} with $\eignorm{}=\sqrt{\ln N/\ln\ln N}$. As one sees from \fig{fig:bestnormsparse} (where $\sqrt{\ln N/\ln\ln N} \approx 2.31$ for $N = 1.45\times 10^6$) decreasing the norm from the lowest point around $\ln N / \ln \ln N$ forces a very rapid growth of the residual in our fit. Since the logarithm and the root of it are very slowly increasing functions one needs extremely large $N$ to distinguish between them, however still the lack of samples may influence the residual dependency on the normalization constant.

To summarise, $\lev{}$ (which is not bigger than $\dmax{}$) and $\dmax{}$ itself are two random variables which are equal $\sqrt{\ln N/\ln\ln N}$ and $\ln N/\ln\ln N$ in the thermodynamic limit. However, since we are interested in statistics of $\lev{}$ for a finite $N$ and finite set of samples, we may meet hypothetically the situation where $\lev{} = \dmax{} = N-1$. Wondering which appropriate norm should be chosen for $\lev{}$, we address to \eq{eq:prop1ineq}. Choosing the desirable $\delta$, one finds such $N$ that $\prob\left\{\exists i: d_i \geq (1+\delta)\ln N /\ln \ln N\right\}$ would be much less than $1/\text{number of samples}$. That would give an effective upper bound for $\dmax{}$ and, as a result, for $\lev{}$.

\section{Conclusion}

We have analyzed semi-analytically -- semi-numerically the statistics of eigenvalues in the vicinity of the spectral boundary of large sparse random adjacency matrices with the bimodal distribution of matrix elements, $a_{ij}$, i.e.: $P(a_{ij}=1)=p$ and $P(a_{ij}=0)=1-p$ where $p=c/N$ and $c$ is close to 1.   

We have shown that the Gumbel distribution emerges for the largest eigenvalue of tridiagonal matrices (see \eq{eq:tridiaglevCDF}), which are adjacency matrices of linear subgraphs (see \eq{eq:quanlinER}). Based on this anslysis we have proposed an ansatz for the distribution of the largest eigenvalue in the ensemble of sparse adjacency matrices and have checked numerically its validity using the variational approach. Specifically, we have demonstrated that if the limiting value of the largest eigenvalue in the ensemble of tridiagonal matrices is replaced by $\eignorm{}\approx \frac{\ln N}{\ln \ln N}$, then the value $\pi/\arccos\left(\lev{}/ \eignorm{}\right)$ still possesses the Gumbel distribution \eq{eq:gumbel} for the ensemble of sparse matrices at $N\gg 1$ at least slightly above of the percolation threshold. 

In the sparse regime the {\it extremal value statistics} (like the ``longest success run'' for the ensemble of tridiagonal matrices) matters and the Lifshitz tail of the spectral density $\rho(\lambda)\sim e^{-g(p)/\sqrt{|\lev{}-\lambda|}}$ close to $\lev{}$ ensures that the finite-size corrections to the largest eigenvalue have logarithmic behavior (see \eq{eq:11}):
\be
\left|\lambda_{\rm max}(\infty) - \lambda_{\rm max}(N)\right| \sim  \ln^{-2} N
\label{eq:sparse1}
\ee

The last question which we would like to comment concerns the dependence of the gap between the largest eigenvalue and the spectral boundary of the main zone (see \fig{fig:01}a) as a function of $p$. In the dense regime ($p=O(1)$) the largest eigenvalue $\lev{}$ is detached from the boundary of the semicircle by a gap of order of $pN$ as it follows from the Frobenius theorem -- compare \eq{eq:boundary} and \eq{eq:04}. On the other hand, at the percolation threshold, $p=1/N$, the largest eigenvalue $\lev{}$ coincides with the boundary of the main zone meaning that the gap between $\lev{}$ and $\levbound{}$ is closed. From \fig{fig:levofp} one can see that with decreasing $p$ from $p=O(1)$ towards the percolation threshold, $p=1/N$, the distance $|\lev{}-\levbound{}|$ shrinks and below $p=\frac{\ln N}{N}$ becomes of order of the distance between neighboring eigenvalues in the main zone (i.e. $\lev{}$ and $\levbound{}$ become indistinguishable). Let us note that the distance $|\lev{}-\levbound{}|$ shrinks slower than the distance from the largest eigenvalue of the second giant component which has maximum at the percolation point and is nullified (in average) before reaching $p_c^*=\frac{\ln N}{N}$. 

\begin{figure}[ht]
\centering
\includegraphics[width=0.7\linewidth]{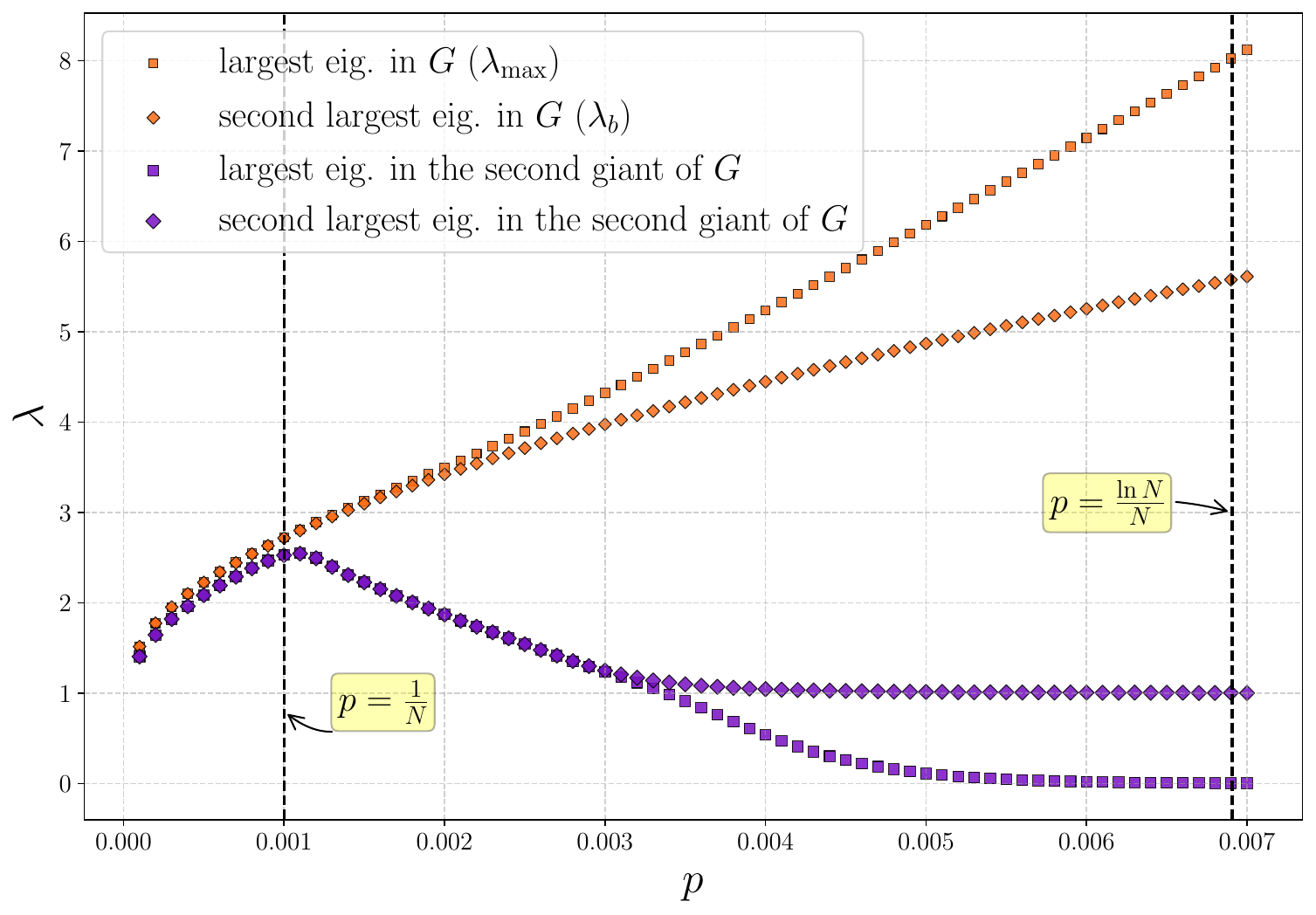}
\caption{\label{fig:levofp}
The largest eigenvalue ($\lev{}$), the second largest eigenvalue (boundary of the main zone, $\levbound{}$) and two largest eigenvalues of the second giant component in $G\sim \fancyG(N, p)$, where $N=1000$. Each value is averaged over $10^5$ samples of $G$. After $p=1/N$ the giant component starts to crowd out all other subgraphs, which is seen in the steep drop of the largest eigenvalue of the second giant component. $p=\ln N/N$ is the probability point where the whole graph almost sure becomes connected. Presence of the largest eigenvalue of the second giant component till this point is the sign of sporadic samples with subgraphs of small sizes.
}
\end{figure}

\begin{acknowledgments}
We are grateful to Alexander Gorsky for valuable discussions on different stages of the work. KP acknowledges hospitality of LPTMS (CNRS-Universit\'e Paris-Saclay) and Institute Curie (Paris) where a part of the work has been done.
\end{acknowledgments}

\bibliography{biblio}

\end{document}